\documentclass[11pt]{article}
\usepackage{microtype}
\usepackage{fullpage,amsthm,amssymb,amsmath}
\usepackage{graphicx,algorithmic,algorithm}
\usepackage{enumerate}
\usepackage{tikz}
\usetikzlibrary{shapes}
\usepackage{hyperref}
\usepackage{complexity}
\usepackage[normalem]{ulem}
\usepackage{authblk}

\usepackage{xcolor}
\definecolor{dark-red}{rgb}{0.4,0.15,0.15}
\definecolor{dark-blue}{rgb}{0.15,0.15,0.4}
\definecolor{medium-blue}{rgb}{0,0,0.5}
\definecolor{gray}{rgb}{0.5,0.5,0.5}

\hypersetup{
    pdftitle=   {Generalized distance domination problems and their complexity on graphs of bounded mim-width},
   pdfauthor=  {Lars Jaffke, O-joung Kwon, Torstein J. F. Str\o{}mme, and Jan Arne Telle}, 
    colorlinks, linkcolor={dark-red},
    citecolor={dark-blue}, urlcolor={medium-blue}
}

\title{Generalized distance domination problems and their complexity on graphs of bounded mim-width}

\author[1]{Lars Jaffke\thanks{Supported by the Bergen Research Foundation (BFS).}}

\author[2]{O-joung Kwon\thanks{Partially supported by the European Research Council (ERC) under the European Union's Horizon 2020 research and innovation programme (ERC consolidator grant DISTRUCT, agreement No. 648527).}}

\author[1]{Torstein J.\,F.\, Str\o{}mme}

\author[1]{Jan Arne Telle}

\affil[1]{Department of Informatics, University of Bergen, Norway. \protect\\ \texttt{\{lars.jaffke, torstein.stromme, jan.arne.telle\}@uib.no}}
\affil[2]{Department of Mathematics, Incheon National University, South Korea. \protect\\ \texttt{ojoungkwon@gmail.com}}

\theoremstyle{plain}
\newtheorem{theorem}{Theorem}

\newtheorem{proposition}[theorem]{Proposition}

\theoremstyle{remark}
\newtheorem{corollary}[theorem]{Corollary}

\newtheorem{observation}[theorem]{Observation}
\newtheorem{claim}{Claim}[section]
\newtheorem{remarknr}[claim]{Remark}
\newtheorem{subclaim}{Claim}[theorem]

\newtheorem{nestedobservation}[subclaim]{Observation}

\theoremstyle{definition}
\newtheorem{definition}[theorem]{Definition}

\newcommand{\dist}{\textsc{dist}}

\newcommand{\Wone}{\text{$\W[1]$}}

\usepackage{xspace}

\DeclareMathOperator{\operatorClassFPT}{FPT}
\newcommand{\classFPT}{\ensuremath{\operatorClassFPT}}
\DeclareMathOperator{\operatorClassNP}{NP}
\newcommand{\classNP}{\ensuremath{\operatorClassNP}}

\newcommand{\LCVPshort}{\textsc{LCVP}\xspace}

\newcommand{\eg}{e.\,g.\@\xspace}

\newcommand*{\defeq}{\mathrel{\vcenter{\baselineskip0.5ex \lineskiplimit0pt
                     \hbox{\scriptsize.}\hbox{\scriptsize.}}}%
                     =}

\newenvironment{claimproof}{\begin{proof}\renewcommand{\qedsymbol}{\claimqed}}{\end{proof}\renewcommand{\qedsymbol}{\plainqed}}

\let\plainqed\qedsymbol

\newcommand\card[1]{\left| #1 \right|}
\newcommand\cardd[1]{\left|\left| #1 \right|\right|}

\newcommand\cM{\mathcal{M}}

\newcommand\cO{\mathcal{O}}

\newcommand\cS{\mathcal{S}}

\newcommand\cV{\mathcal{V}}

\newcommand\bN{\mathbb{N}}

\newcommand\fB{\mathfrak{B}}
\newcommand\fC{\mathfrak{C}}

\newcommand\maxprob{\textsc{Max}}
\newcommand\minprob{\textsc{Min}}

\begin{document}

\maketitle

\begin{abstract}
	We generalize the family of $(\sigma, \rho)$-problems and locally checkable vertex partition problems to their distance versions, which naturally captures well-known problems such as distance-$r$ dominating set and distance-$r$ independent set. We show that these distance problems are \XP{} parameterized by the structural parameter {\em mim-width}, and hence polynomial on graph classes where mim-width is bounded and quickly computable, such as $k$-trapezoid graphs, Dilworth $k$-graphs, (circular) permutation graphs, interval graphs and their complements, convex graphs and their complements, $k$-polygon graphs, circular arc graphs, complements of $d$-degenerate graphs, and $H$-graphs if given an $H$-representation. To supplement these findings, we show that many classes of (distance) $(\sigma, \rho)$-problems are \Wone{}-hard parameterized by mim-width + solution size.
\end{abstract}

\section{Introduction}
Telle and Proskurowski~\cite{telle1997algorithms} defined the $(\sigma,\rho)$-domination problems, and the more general \emph{locally checkable vertex partitioning} problems (\LCVPshort{}). In $(\sigma,\rho)$-domination problems, feasible solutions are vertex sets with constraints on how many neighbours each vertex of the graph has in the set.
The framework generalizes important and well-studied problems such as \textsc{Maximum Independent Set} and \textsc{Minimum Dominating Set},
as well as \textsc{Perfect Code}, \textsc{Minimum subgraph with minimum degree $d$} and a multitude of other problems. See Table \ref{tab1}. Bui-Xuan, Telle and Vatshelle~\cite{BUIXUAN201366} showed that $(\sigma, \rho)$-domination and locally checkable vertex partitioning problems can be solved in time \XP{} parameterized by \emph{mim-width}, if we are given a corresponding decomposition tree. Roughly speaking, the structural parameter mim-width measures how easy it is to decompose a graph along vertex cuts inducing a bipartite graph with small maximum induced matching size~\cite{vatshelle2012new}.

In this paper, we consider distance versions of problems related to independence and domination, like \textsc{Distance-$r$ Independent Set} and \textsc{Distance-$r$ Dominating Set}. 
The \textsc{Distance-$r$ Independent Set} problem, also studied under the names \textsc{$r$-Scattered Set} and \textsc{$r$-Dispersion} (see e.g.~\cite{BMT17} and the references therein), asks to find a set of at least $k$ vertices whose vertices have pairwise distance strictly longer than $r$. Agnarsson et al.~\cite{AGNARSSON2003} pointed out that it is identical to the original \textsc{Independent Set} problem on the $r$-th power graph $G^{r}$ of the input graph $G$, and also showed that for fixed $r$, it can be solved in linear time for interval graphs, and circular arc graphs.
The \textsc{Distance-$r$ Dominating Set} problem was introduced by Slater~\cite{SLATER1976} and Henning et al.~\cite{HENNING1991}. They also discussed that 
it is identical to solve the original \textsc{Dominating Set} problem on the $r$-th power graph. Slater presented a linear-time algorithm to solve \textsc{Distance-$r$ Dominating Set} problem on forests.

We generalize all of the $(\sigma,\rho)$-domination and \LCVPshort{} problems to their distance versions, which naturally captures  \textsc{Distance-$r$ Independent Set} and \textsc{Distance-$r$ Dominating Set}. Where the original problems put constraints on the size of the immediate neighborhood of a vertex, we consider the constraints to be applied to the ball of radius $r$ around it. Consider for instance the \textsc{Minimum Subgraph with Minimum Degree $d$} problem;
where the original problem is asking for the smallest (number of vertices) subgraph of minimum degree $d$, we are instead looking for the smallest subgraph such that for each vertex there are at least $d$ vertices at distance at least $1$ and at most $r$.
In the \textsc{Perfect Code} problem, the target is to choose a subset of vertices such that each vertex has exactly one chosen vertex in its closed neighbourhood. In the distance-$r$ version of the problem, we replace the closed neighbourhood by the closed $r$-neighbourhood. This problem is known as \textsc{Perfect $r$-Code}, and was introduced by Biggs~\cite{BIGGS1973289} in 1973. Similarily, for every problem in Table~\ref{tab1} its distance-$r$ generalization either introduces a new problem or is already well-known.

We show that all these distance problems are \XP{} parameterized by mim-width if a decomposition tree is given. The main result of the paper is of structural nature, namely that for any positive integer $r$ the mim-width of a graph power $G^r$ is at most twice the mim-width of $G$. It follows that we can reduce the distance-$r$ version of a $(\sigma,\rho)$-domination problem to its non-distance variant by taking the graph power $G^r$, whilst preserving small mim-width.

The downside to showing results using the parameter mim-width, is that we do not know an \XP{} algorithm computing mim-width. 
Computing a decomposition tree with optimal mim-width is \classNP{}-complete in general and \Wone{}-hard parameterized by itself. Determining the optimal mim-width is not in $\APX$ unless $\NP = \ZPP$, making it unlikely to have a polynomial-time constant-factor approximation algorithm~\cite{DBLP:journals/tcs/SaetherV16}, but saying nothing about an \XP{} algorithm.

However, for several graph classes we are able to find a decomposition tree of constant mim-width in polynomial time, using the results of Belmonte and Vatshelle~\cite{BELMONTE201354}. These include; permutation graphs, convex graphs and their complements, interval graphs and their complements (all of which have {\em linear} mim-width 1); (circular $k$-) trapezoid graphs, circular permutation graphs, Dilworth-$k$ graphs, $k$-polygon graphs, circular arc graphs and complements of $d$-degenerate graphs. Fomin, Golovach and Raymond~\cite{FGR17} show that we can find linear decomposition trees of constant mim-width for the very general class of $H$-graphs, see Definition~\ref{def:hgraph}, in polynomial time, {\em given}\footnote{We would like to remark that it is $\NP$-complete to decide whether a graph is an $H$-graph whenever $H$ is not a cactus~\cite{CTVZ17}.} an $H$-representation of the input graph. For all of the above graph classes, our results imply that the distance-$r$ $(\sigma, \rho)$-domination and \LCVPshort{} problems become polynomial time solvable.

Graphs represented by intersections of objects in some model are often closed under taking powers.
For instance, interval graphs, and generally $d$-trapezoid graphs~\cite{FLOTOW1995,AGNARSSON2003}, circular arc graphs~\cite{RAYCHAUD1992,AGNARSSON2003},  and leaf power graphs (by definition) are such graphs. 
We refer to \cite[Chapter 10.6]{BRANDSTADT1999} for a survey of such results.
For these classes, we already know that the distance-$r$ version of a $(\sigma,\rho)$-domination problem can be solved in polynomial time.
However, this closure property does not always hold; for instance, permutation graphs are not closed under taking powers.
Our result provides that to obtain such algorithmic results, we do not need to know that these classes are closed under taking powers; it is sufficient to know that classes have bounded mim-width. To the best of our knowledge, for the most well-studied distance-$r$ $(\sigma, \rho)$-domination problem, \textsc{Distance-$r$ Dominating Set}, we obtain the first polynomial time algorithms on Dilworth $k$-graphs, convex graphs and their complements, complements of interval graphs, $k$-polygon graphs, $H$-graphs (given an $H$-representation of the input graph), and complements of $d$-degenerate graphs.

The natural question to ask after obtaining an \XP{} algorithm, is whether we can do better, \eg can we show that for all fixed $r$, the distance-$r$ $(\sigma, \rho)$-domination problems are in \FPT{}. Fomin et al.~\cite{FGR17} answered this in the negative by showing that (the standard, i.e. distance-$1$ variants of) \textsc{Maximum Independent Set}, \textsc{Minimum Dominating Set} and \textsc{Minimum Independent Dominating Set} problems are \Wone{}-hard parameterized by (linear) mim-width + solution size. 
We modify their reductions to extend these results to several families of $(\sigma, \rho)$-domination problems, including the maximization variants of {\sc Induced Matching}, {\sc Induced $d$-Regular Subgraph} and {\sc Induced Subgraph of Max Degree $\le d$}, the minimization variants of {\sc Total Dominating Set} and {\sc $d$-Dominating Set} and both the maximization and the minimization variant of {\sc Dominating Induced Matching}.

The remainder of the paper is organized as follows. In Section~\ref{sec:intro-sigmarho} we introduce the $(\sigma, \rho)$ problems and define their distance-$r$ generalization. In Section~\ref{sec:intro-mimw} we introduce mim-width, and state previously known results. In Section~\ref{sec:mimwpower} we show that the mim-width of a graph grows by at most a factor 2 when taking (arbitrary large) powers and give algorithmic consequences. We discuss LCVP problems, their distance-$r$ versions and algorithmic consequences regarding them in Section~\ref{sec:lcvp} and in
Section~\ref{sec:lowerbounds} we present the above mentioned lower bounds. Finally, we give some concluding remarks in Section~\ref{sec:conclusion}. 
Some notational conventions are given in the appendix.

\section{Distance-$r$ $(\sigma, \rho)$-Domination Problems}\label{sec:intro-sigmarho}
Let $\sigma$ and $\rho$ be finite or co-finite subsets of the natural numbers $\sigma, \rho \subseteq \bN$. For a graph $G$, a vertex set $S\subseteq V(G)$ is a $(\sigma, \rho)$-dominator if
\begin{itemize}
    \item for each vertex $v\in S$ it holds that $|N(v) \cap S| \in \sigma$, and
    \item for each vertex $v\in V(G)\setminus S$ it holds that $|N(v) \cap S| \in \rho$.
\end{itemize}
\noindent For instance, a $(\{0\}, \mathbb{N})$-set is an independent set as there are
no edges inside of the set, and we do not care about adjacencies between
$S$ and $V(G)\setminus S$.
For another example, a $(\mathbb{N}, \mathbb{N}^+)$-set is a dominating
set as we require that for each vertex in $V(G)\setminus S$, it has at
least one neighbor in $S$.

There are 3 types of $(\sigma, \rho)$-domination problems; minimization, maximization and existence. We denote the problem of finding a minimum $(\sigma, \rho)$-dominator as the \textsc{Min}-$(\sigma, \rho)$ problem. Similarily, \textsc{Max}-$(\sigma, \rho)$ denotes the maximization problem, and $\exists\text{-}(\sigma, \rho)$ denotes the existence problem. Many well-studied problems fit into this framework, see Table \ref{tab1} for examples.

\begin{table}[tb]
    \centering
    \begin{tabular}{|l|l|l|l|}   \hline
    $\sigma$ & $\rho$ & $d$ & Standard name \\
    \hline \hline

    $\{0\}$ & $\mathbb{N}$ & 1 & Independent set $*$ \\
    $\mathbb{N}$ & $\mathbb{N}^+$ & 1 & Dominating set $**$ \\
    $\{0\}$ & $\mathbb{N}^+$ &  1 & Maximal Independent set $**$  \\
    $\mathbb{N}^+$ & $ \mathbb{N}^+$ &  1 & Total Dominating set $\star \star$  \\

    $\{0\}$ & $\{0,1\}$ & 2 &Strong Stable set or 2-Packing \\
    $\{0\}$ & $\{1\}$ &  2 &Perfect Code or Efficient Dom. set  \\
    $\{0,1\}$ & $\{0,1\}$ & 2 & Total Nearly Perfect set \\
    $\{0,1\}$ & $\{1\}$ &  2 & Weakly Perfect Dominating set \\
    $\{1\}$ & $\{1\}$ & 2 & Total Perfect Dominating set \\
    $ \{1\}$ & $\mathbb{N}$ & 2 & Induced Matching $\star$ \\
    $\{1\}$ & $\mathbb{N}^+$ & 2 & Dominating Induced Matching $\star$, $\star \star$ \\
    $\mathbb{N}$ & $\{1\}$ & 2 & Perfect Dominating set  \\

    $\mathbb{N}$ &
    $\{d,\hspace{-0.05cm}d\hspace{-0.01cm}+\hspace{-0.01cm}1,\hspace{-0.05cm}...\hspace{-0.05cm}\}$
    & $d$&  $d$-Dominating set $\star \star$  \\

    $\{d\}$ & $\mathbb{N}$ &
    $\hspace{-0.01cm}d\hspace{-0.01cm}+\hspace{-0.01cm}1$ &
    Induced $d$-Regular Subgraph $\star$  \\

    $\{d,\hspace{-0.05cm}d\hspace{-0.01cm}+\hspace{-0.01cm}1,\hspace{-0.05cm}...\hspace{-0.05cm}\}$
    & $\mathbb{N}$ & $d$ &  Subgraph of Min Degree $\geq d$  \\

    $\{0,1,\hspace{-0.05cm}...,d\}\hspace{-0.01cm}$ & $\mathbb{N}$ &
    $d\hspace{-0.01cm}+\hspace{-0.01cm}1$ &
    Induced Subg. of Max Degree $\leq d$  $\star$ \\

    \hline
    \end{tabular}
    \caption{Some vertex subset properties expressible as $(\sigma,
    \rho)$-sets, with $\mathbb{N}=\{0,1,...\}$ and $\mathbb{N}^+=\{1,2,...\}$.
    Column $d$ shows $d=\max( d(\sigma), d(\rho))$. For each problem, at least one of the minimization, the maximization and the existence problem is NP-complete. For problems marked with $\star$ (resp., $\star \star$), \Wone{}-hardness of the maximization (resp., minimization) problem parameterized by mim-width + solution size is shown in the present paper. For problems marked with $*$ (resp., $* *$) the \Wone{}-hardness of maximization (resp., minimization) in the same parameterization was shown by Fomin et al.~\cite{FGR17}.}
    \label{tab1}
\end{table}
The \emph{$d$-value} of a distance-$r$ $(\sigma, \rho)$ problem is a constant which will ultimately affect the runtime of the algorithm. For a set $\mu \subseteq \bN$, the value $d(\mu)$ should be understood as the highest value in $\bN$ we need to enumerate in order to describe $\mu$. Hence, if $\mu$ is finite, it is simply the maximum value in $\mu$, and if $\mu$ is co-finite, it is the maximum natural number \emph{not} in $\mu$ ($1$ is added for technical reasons).
\begin{definition}[{$d$-value}] \label{def:dvalue}
  Let $d(\bN) = 0$. For every non-empty finite or co-finite set $\mu \subseteq \bN$, let $d(\mu) = 1+\min(\max\{x \mid x \in \mu\}, \max\{x \mid x \in \bN \setminus \mu\})$.
\end{definition}
For a given distance-$r$ $(\sigma, \rho)$ problem $\Pi_{\sigma, \rho}$, its $d$-value is defined as $d(\Pi_{\sigma, \rho}) \defeq \max\{d(\sigma), d(\rho)\}$, see column $d$ in Table~\ref{tab1}.

\section{Mim-width and Applications}\label{sec:intro-mimw}
\emph{Maximum induced matching width}, or \emph{mim-width} for short, was introduced in the Ph.\,D.\, thesis of Vatshelle~\cite{vatshelle2012new}, used implicitly by Belmonte and Vatshelle~\cite{BELMONTE201354}, and is a structural graph parameter described over \emph{decomposition trees} (sometimes called \emph{branch decompositions}), similar to graph parameters such as \emph{rank-width} and \emph{module-width}. Decomposition trees naturally appear in divide- and conquer -style algorithms where one recursively partitions the pieces of a problem into two parts. When the algorithm is at the point where it combines solutions of its subproblems to form a full solution, the structure of the cuts are (unsurprisingly) important to the runtime; this is especially true of dynamic programming when one needs to store multiple sub-solutions at each intermediate node. We will briefly introduce the necessary machinery here, but for a more comprehensive introduction we refer the reader to~\cite{vatshelle2012new}.

A graph of maximum degree at most $3$ is called {\em subcubic}.
A {\em decomposition tree} for a graph $G$ is a pair $(T, \delta)$ where $T$ is a subcubic tree and $\delta : V(G) \to L(T)$ is a bijection between the vertices of $G$ and the leaves of $T$. Each edge $e\in E(T)$ naturally splits the leaves of the tree in two groups depending on their connected component when $e$ is removed. In this way, each edge $e\in E(T)$ also represent a partition of $V(G)$ into two partition classes $A_e$ and $\overline{A}_e$.
One way to measure the cut structure is by the \emph{maximum induced matching} across a cut of $(T, \delta)$. A set of edges $M$ is called an {\em induced matching} if no pair of edges in $M$ shares an endpoint and if the subgraph induced by the endpoints of $M$ does not contain any additional edges.
\begin{definition}[mim-width]
    Let $G$ be a graph, and let $(T, \delta)$ be a decomposition tree for $G$. For each edge $e \in E(T)$ and corresponding partition of the vertices $A_e, \overline{A}_e$, we let $cutmim_G(A_e, \overline{A}_e)$ denote the size of a maximum induced matching of the bipartite graph on the edges crossing the cut. Let the {\em mim-width of the decomposition tree} be
    $$mimw_G(T, \delta) = \max_{e \in E(T)}\{cutmim(A_e, \overline{A}_e)\}$$
    The {\em mim-width of the graph $G$}, denoted $mimw(G)$, is the minimum value of $mimw_G(T, \delta)$ over all possible decompositions trees $(T, \delta)$. The {\em linear mim-width of the graph $G$} is the minimum value of $mimw_G(T, \delta)$ over all possible decompositions trees $(T, \delta)$ where $T$ is a caterpillar.
\end{definition}

In previous work, Bui-Xuan et al.~\cite{BUIXUAN201366} and Belmonte and Vatshelle~\cite{BELMONTE201354} showed that all $(\sigma, \rho)$ problems can be solved in time $n^{\cO(w)}$ where $w$ denotes the mim-width of a decomposition tree that is provided as part of the input. More precisely, they show the following.\footnote{We would like to remark that the original results in~\cite{BUIXUAN201366} are stated in terms of the number of $d$-neighborhood equivalence classes across the cuts in the decomposition tree ($nec_d(T, \delta)$) giving a runtime of $n^4 \cdot nec_d(T, \delta)^c$ (where $c = 2$ if the given decomposition is a caterpillar and $c = 3$ otherwise). In~\cite[Lemma 2]{BELMONTE201354}, Belmonte and Vatshelle show that $nec_d(T, \delta) \le n^{d \cdot mimw_G(T, \delta)}$.}

\begin{proposition}[{\cite{BELMONTE201354,BUIXUAN201366}}] \label{prop:sigma:rho-algorithm-nondistance-minmw}
	There is an algorithm that given a graph $G$ and a decomposition tree $(T, \delta)$ of $G$ with $w \defeq mimw_G(T, \delta)$ solves each $(\sigma, \rho)$ problem $\Pi$ with $d \defeq d(\Pi)$
		\begin{enumerate}[(i)]
			\item in time $\cO(n^{4 + 2d \cdot w})$, if $T$ is a caterpillar, and
			\item in time $\cO(n^{4 + 3d \cdot w})$, otherwise.
		\end{enumerate}
\end{proposition}

\section{Mim-width on Graph Powers}\label{sec:mimwpower}
\begin{definition}[Graph power]
  Let $G = (V, E)$ be a graph. Then the $k$-th power of $G$, denoted $G^k$, is a graph on the same vertex set where there is an edge between two vertices if and only if the distance between them is at most $k$ in $G$. Formally, $V(G^k) = V(G)$ and $E(G^k) = \{uv \mid \dist_G(u, v) \leq k \}$.
\end{definition}
\begin{theorem}\label{thm:mimwpower}
  For any graph $G$ and positive integer $k$, $\text{mimw}(G^k) \leq 2 \cdot \text{mimw}(G)$.
\end{theorem}
\begin{proof}
  Assume that there is a decomposition tree of mim-width $w$ for the graph $G$. We show that the same decomposition tree has mim-width at most $2w$ for $G^k$.
  
  We consider a cut $A, \overline{A}$ of the decomposition tree. Let $M$ be a maximum induced matching across the cut for $G^k$. To prove our claim, it suffices to construct an induced matching across the cut $M'$ in $G$ such that $|M'| \geq \frac{|M|}{2}$. 
  
  We begin by noticing that for an edge $uv \in M$, the distance between $u$ and $v$ is at most $k$ in $G$. For each such edge $uv \in M$, we let $P_{uv}$ denote some shortest path between $u$ and $v$ in $G$ (including the endpoints $u$ and $v$).
  \begin{subclaim} \label{claim:disjointpaths}
    Let $uv, wx \in M$ be two distinct edges of the matching. Then $P_{uv}$ and $P_{wx}$ are vertex disjoint.
  \end{subclaim}
  \begin{claimproof}
    We may assume that $u, w \in A$ and $v, x \in \overline{A}$. 
    Now assume for the sake of contradiction there exists a vertex $y \in P_{uv} \cap P_{wx}$. Because both paths have length at most $k$, we have that $\dist_G(u,y) + \dist_G(y, v) \leq k$, and $\dist_G(w,y) + \dist_G(y, x) \leq k$. Adding these together, we get
    $$\dist_G(u,y) + \dist_G(y, v) + \dist_G(w,y) + \dist_G(y, x) \leq 2k.$$
    \noindent Since $uv$ and $wx$ are both in $M$, there can not exist edges $ux$ and $wv$ in $G^k$. Hence, their distance in $G$ is strictly greater than $k$, i.e.~$\dist_G(u,y) + \dist_G(y, x) \ge \dist_G(u, x) > k$, and $\dist_G(w,y) + \dist_G(y, v) > k$. Putting these together, we obtian our contradiction:
    $$\dist_G(u,y) + \dist_G(y, x) + \dist_G(w,y) + \dist_G(y, v) > 2k$$
    This concludes the proof of the claim.
  \end{claimproof}
  Our next observation is that for each $uv \in M$, the path $P_{uv}$ starts (without loss of generality) in $A$, and ends in $\overline{A}$. There must hence exist at least one point at which the path cross from $A$ to $\overline{A}$. For each $uv \in M$, we can thus safely let $u'v' \in E(P_{uv})$ denote an edge in $G$ such that $u' \in A$ and $v' \in \overline{A}$.

  We plan to construct our matching $M'$ by picking a subset of such edges. However, we can not simply take all of them, since some pairs may be incompatible in the sense that they will not form an induced matching across the cut $A, \overline{A}$. We examine the structures that arise when two such edges $u'v'$ and $w'x'$ are incompatible, and can not both be included in the same induced matching across the cut. For easier readability, we let $\alpha_d$ be a shorthand notation for $\dist_G(\alpha, \alpha')$ for $\alpha \in \{u, v, w, x\}$.
  \begin{figure}
    \centering

\begin{tikzpicture}[scale=.8]
  \usetikzlibrary{decorations.pathmorphing}
  \tikzset{snake it/.style={decorate, decoration={snake, segment length=4, amplitude=1}}}
  
  \begin{scope}[every node/.style={circle, draw, minimum width=.9cm}]
    \node (u) at (0, 0) {$u$};
    \node (v) at (0, -2) {$v$};
    \node (up) at (3, 0) {$u'$};
    \node (vp) at (3, -2) {$v'$};
  
    \node (w) at (9, 0) {$w$};
    \node (x) at (9, -2) {$x$};
    \node (wp) at (6, 0) {$w'$};
    \node (xp) at (6, -2) {$x'$};
  \end{scope}
  
  \node at (-1.8, -0.5) {$A$};
  \node at (-1.8, -1.5) {$\overline{A}$};
  \draw[dotted] (-2, -1) -- (10, -1);
  
  \draw[dashed](u) to (v);
  \path[-](u) edge [snake it, bend left=10] node[fill=white] {$u_d$} (up);
  \path[-](v) edge [snake it, bend right=10] node[fill=white] {$v_d$} (vp);
  \draw (up) to (vp);

  \draw[dashed](w) to (x);
  \path[-] (w) edge [snake it, bend right=10] node[fill=white] {$w_d$} (wp);
  \path[-] (x) edge [snake it, bend left=10] node[fill=white] {$x_d$} (xp);
  \draw (wp) to (xp);
  
  \draw (up) to (xp);

\end{tikzpicture}
    \caption{Structure of two paths $P_{uv}$ and $P_{wx}$ when the edge $u'x'$ exists in $G$. Dashed edges appear in $G^k$, solid edges appear in $G$, squiggly lines are (shortest) paths existing in $G$ (possibly of length 0, and possibly crossing back and forth across the cut).}  \label{fig:neighbourpaths}
  \end{figure}
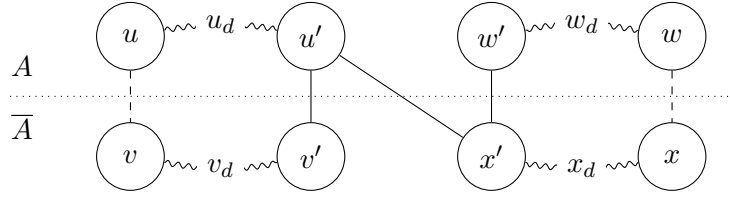  
  \begin{subclaim} \label{claim:neighbourpaths}
    Let $uv, wx \in M$ be two distinct edges of $M$ and let $u'v'$ and $w'x'$ be edges on the shortest paths as defined above. 
    If there is an edge $u'x' \in E(G)$, then all of the following hold. See Figure~\ref{fig:neighbourpaths}.
    \begin{enumerate}[(a)]
      \item $u_d + x_d = k$
      \item $u_d + v_d = w_d + x_d = k - 1$
      \item $w_d = u_d - 1$
    \end{enumerate}
  \end{subclaim}
  \begin{claimproof}
%
  \textbf{(a)} Since $ux$ is not an edge in $G^k$, the distance between $u$ and $x$ must be at least $k+1$ in $G$, and so $u_d + x_d$ must be at least $k$. It remains to show that $u_d + x_d \leq k$ for equality to hold. Similarily to the proof of Claim~\ref{claim:disjointpaths}, we know that $P_{uv}$ and $P_{wx}$ both are of length at most $k$. We get
  \begin{equation}\label{eq:sum_all_upperbound}
  u_d + v_d + w_d + x_d \leq 2k - 2
  \end{equation}
  \noindent The $-2$ at the end is because we do not include the length contributed by edges $u'v'$ and $w'x'$ in our sum. Now assume for the sake of contradiction that $u_d + x_d \geq k+1$. Then we get that
  $$v_d + w_d \leq 2k - 2 - k - 1 = k-3$$
  \noindent Because $\dist_G(v', w') \leq 3$ (follow the edges $u'v'\to u'x' \to w'x'$), this implies that $\dist_G(v, w) \leq k$, and the edge $vw$ would hence exist in $G^k$. This contradicts that $uv$ and $wx$ were both in the same induced matching $M$.
  
  \textbf{(b)} Assume for the sake of contradiction that $u_d + v_d \leq k - 2$. Then, rather than Equation~\ref{eq:sum_all_upperbound}, we get the following bound
  $$u_d + v_d + w_d + x_d \leq 2k - 3$$
  \noindent By (a) we know that $u_d + x_d = k$, so by a similar argument as above we get that $v_d + x_d \leq k - 3$, obtaining a contradiction. An anolgous argument holds for $w_d + x_d$.
  
  \textbf{(c)} This follows immidiately by substituting (a) into (b).
  \end{claimproof}
  We will now construct our induced matching $M'$. We construct two candidates for $M'$, and we will pick the biggest one. First, we construct $M'_0$ by including $u'v'$ for each edge $uv \in M$ where $\dist_G(u, u')$ is even. Symetrically, $M'_1$ is constructed by including $u'v'$ if $\dist_G(u, u')$ is odd. Clearly, at least one of $M'_0, M'_1$ contains $\geq \frac{|M|}{2}$ egdes. It remains to show that $M'$ indeed forms an induced matching across the cut $A, \overline{A}$ in $G$.
  
  Consider two distinct edges $u'v'$ and $w'x'$ from $M'$. By Claim~\ref{claim:disjointpaths}, the two edges are vertex disjoint. If there is an edge violating that $u'v'$ and $w'x'$ are both in the same induced matching, it must be either $u'x'$ or $v'w'$. Without loss of generality we may assume it is an edge of the type $u'x'$. By Claim~\ref{claim:neighbourpaths} (c), we then have that the parities of $\dist_G(u, u')$ and $\dist_G(w, w')$ are different. But by how $M'$ was constructed, this is not possible. This concludes the proof.
\end{proof}

\begin{observation} \label{obs:distancenbhood_equiv_powergraph}
  For a positive integer $r$, a graph $G$ and a vertex $u\in V(G)$, the $r$-neighbourhood of $u$ is equal to the neighbourhood of $u$ in $G^r$, i.e.~$N^r_G(u) = N_{G^r}(u)$.
\end{observation}
\noindent The observation above shows that solving a distance-$r$ $(\sigma, \rho)$ problem on $G$ is the same as solving the same standard distance-1 variation of the problem on $G^r$. Hence, we may reduce our problem to the standard version by simply computing the graph power. Combining Theorem~\ref{thm:mimwpower} with the algorithms provided in Proposition~\ref{prop:sigma:rho-algorithm-nondistance-minmw}, we have the following consequence.
\begin{corollary}\label{cor:dist:r:mim}
	There is an algorithm that for all $r \in \bN$, given a graph $G$ and a decomposition tree $(T, \delta)$ of $G$ with $w \defeq mimw_G(T, \delta)$ solves each distance-$r$ $(\sigma, \rho)$ problem $\Pi$ with $d \defeq d(\Pi)$
	\begin{enumerate}[(i)]
		\item in time $\cO(n^{4 + 4d\cdot w})$, if $T$ is a caterpillar, and\label{cor:dist:r:mim:lin}
		\item in time $\cO(n^{4 + 6d\cdot w})$, otherwise.\label{cor:dist:r:mim:general}
	\end{enumerate}
\end{corollary}
\begin{proof}
    Let $G$ be the input graph and $(T, \delta)$ the provided decomposition tree. We apply the following algorithm:
    \begin{description}
        \item[Step 1.] Compute the graph $G^r$.
        \item[Step 2.] Solve the standard (distance-$1$) version of the problem on $G^r$, providing $(T, \delta)$ as the decomposition tree.
        \item[Step 3.] Return the answer of the algorithm ran in Step 2 without modification.
    \end{description}
    Computing $G^r$ in Step 1 takes at most $\cO(n^3)$ time using standard methods, Step 3 takes constant time. The worst time complexity is hence found in Step 2. By Theorem~\ref{thm:mimwpower}, the mim-width of $(T, \delta)$ on $G^r$ is at most twice that of the same decomposition on $G$. The stated runtime then follows from Proposition~\ref{prop:sigma:rho-algorithm-nondistance-minmw}.
	The correctness of this procedure follows immediately from Observation~\ref{obs:distancenbhood_equiv_powergraph}.
\end{proof}

\section{\LCVPshort{} Problems}\label{sec:lcvp}
A generalization of $(\sigma, \rho)$ problems are the \emph{locally checkable vertex partitioning} (\LCVPshort{}) problems. A {\em degree constraint matrix} $D$ is a $q\times q$ matrix where each entry is a finite or co-finite subset of $\bN$. For a graph $G$ and a partition of its vertices $\cV = \{V_1, V_2, \ldots V_q\}$, we say that it is a $D$-partition if and only if, for each $i, j \in [q]$ and each vertex $v \in V_i$, it holds that $|N(v) \cap V_j| \in D[i,j]$. Empty partition classes are allowed.

For instance, if a graph can be partitioned according to the $3 \times 3$ matrix whose diagonal entries are $\{0\}$ and the non-diagonal ones are $\bN$, then the graph is $3$-colorable.
Typically, the natural algorithmic questions associated with \LCVPshort{} properties are existential.\footnote{Note however that each $(\sigma, \rho)$ problem can be stated as an \LCVPshort{} problem via the matrix $D_{(\sigma,\rho)} = \begin{bmatrix} \sigma & \bN \\ \rho & \bN \end{bmatrix}$, so maximization or minimization of some block of the partition can be natural as well.
}
Interesting problems which can be phrased in such terms include the \textsc{$H$-Covering} and \textsc{Graph $H$-Homomorphism} problems where $H$ is fixed, as well as \textsc{$q$-coloring}, \textsc{Perfect Matching Cut} and more. We refer to~\cite{telle1997algorithms} for an overview.

We generalize \LCVPshort{} properties to their distance-$r$ version, by considering the ball of radius $r$ around each vertex rather than just the immediate neighbourhood.

\begin{definition}[{Distance-$r$ neighbourhood constraint matrix}] \label{def:dist-nbrh-constraint-matrix}
    A distance-$r$ neighbourhood constraint matrix $D$ is a $q\times q$ matrix where each entry is a finite or co-finite subset of $\bN$. For a graph $G$ and a partition of its vertices $\cV = \{V_1, V_2, \ldots V_q\}$, we say that it is a $D$-distance-$r$-partition if and only if, for each $i, j \in [q]$ and each vertex $v \in V_i$, it holds that $|N^r(v) \cap V_j| \in D[i,j]$. Empty partition classes are allowed.
\end{definition}

We say that an algorithmic problem is a \emph{distance-$r$ \LCVPshort{} problem} if the property in question can be described by a distance-$r$ neighbourhood constraint matrix. For example, the distance-$r$ version of a problem such as \textsc{$q$-coloring} can be interpreted as an assignment of at most $q$ colours to vertices of a graph such that no two vertices are assigned the same colour if they are at distance $r$ or closer.

For a given distance-$r$ \LCVPshort{} problem $\Pi$, its $d$-value $d(\Pi)$ is the maximum $d$-value over all the sets in the corresponding neighbourhood constraint matrix.

As in the case of $(\sigma, \rho)$ problems, combining Theorem~\ref{thm:mimwpower} with Observation~\ref{obs:distancenbhood_equiv_powergraph} and the works~\cite{BELMONTE201354,BUIXUAN201366} we have the following result.

\begin{corollary}
	There is an algorithm that for all $r \in \bN$, given a graph $G$ and a decomposition tree $(T, \delta)$ of $G$ with $w \defeq mimw_G(T, \delta)$ solves each distance-$r$ LCVP problem $\Pi$ with $d \defeq d(\Pi)$
	\begin{enumerate}[(i)]
		\item in time $\cO(n^{4 + 4qd \cdot w})$, if $T$ is a caterpillar, and
		\item in time $\cO(n^{4 + 6qd \cdot w})$, otherwise.
	\end{enumerate}
\end{corollary}

\section{Lower Bounds}\label{sec:lowerbounds}
We show that several $(\sigma, \rho)$-problems are \Wone{}-hard parameterized by linear mim-width plus solution size.
Our reductions are based on two recent reductions due to Fomin, Golovach and Raymond~\cite{FGR17} who showed that {\sc Independent Set} and {\sc Dominating Set} are \Wone{}-hard parameterized by linear mim-width plus solution size. In fact they show hardness for the above mentioned problems on {\em $H$-graphs} (the parameter being the number of edges in $H$ plus solution size) which we now define formally.

\begin{definition}[{$H$-Graph}]\label{def:hgraph} Let $X$ be a set and $\cS$ a family of subsets of $X$. The {\em intersection graph} of $\cS$ is a graph with vertex set $\cS$ such that $S, T \in \cS$ are adjacent if and only if $S \cap T \neq \emptyset$.
Let $H$ be a (multi-) graph. We say that $G$ is an {\em $H$-graph} if there is a subdivision $H'$ of $H$ and a family of subsets $\cM \defeq \{M_v\}_{v \in V(G)}$ (called an {\em $H$-representation}) of $V(H')$ where $H'[M_v]$ is connected for all $v \in V(G)$, such that $G$ is isomorphic to the intersection graph of $\cM$.
\end{definition}

All of the hardness results presented in this section are obtained via reductions to the respective problems on $H$-graphs, and the hardness for linear mim-width follows from the following proposition.

\begin{proposition}[Theorem 1 in~\cite{FGR17}]\label{thm:h:graph:mim}
	Let $G$ be an $H$-graph. Then, $G$ has linear mim-width at most $2\cdot \cardd{H}$ and a corresponding decomposition tree can be computed in polynomial time given an $H$-representation of $G$.
\end{proposition}

\subsection{Maximization Problems}
The first lower bound concerns several maximization problems that can be expressed in the $(\sigma, \rho)$ framework. Recall that the {\sc Independent Set} problem can be formulated as \maxprob-$(\{0\}, \bN)$. 
The following result states that a class of problems that generalize the {\sc Independent Set} problem where each vertex in the solution is allowed to have at most some fixed number of $d$ neighbors of the solution, and several variants thereof, is \Wone{}-hard on $H$-graphs parameterized by $\cardd{H}$ plus solution size.

    \newcommand{\probset}{$(\sigma^*, \bN_{\ge x})$}
    \newcommand{\exactprobset}{$(\{d\}, \{d+1, \ldots, d+k\})$}
    \newcommand{\prob}{\maxprob-{\sc \probset ~Domination}}

\begin{theorem}\label{thm:lb:max}
	For any fixed $d \in \bN$ and $x \le d+1$, the following holds. Let $\sigma^* \subseteq \bN_{\le d}$ with $d \in \sigma^*$. Then, \prob ~is \Wone{}-hard on $H$-graphs parameterized by the number of edges in $H$ plus solution size, and the hardness holds even if an $H$-representation of the input graph is given.
\end{theorem}
\begin{proof}
    To prove the theorem, we provide a reduction from {\sc Multicolored Clique} where given a graph $G$ and a partition $V_1, \ldots, V_k$ of $V(G)$, the question is whether $G$ contains a clique of size $k$ using precisely one vertex from each $V_i$ ($i \in [k]$). This problem is known to be $\Wone$-complete~\cite{FHRV09,Pie03}.
    
    Let $(G, V_1, \ldots, V_k)$ be an instance of {\sc Multicolored Clique}. We can assume that $k \ge 2$ and that $\card{V_i} = p$ for $i \in [k]$. If the second assumption does not hold, let $p \defeq \max_{i \in [k]} \card{V_i}$ and add $p - \card{V_i}$ isolated vertices to $V_i$, for each $i \in [k]$. (Note that adding isolated vertices does not change the answer to the problem.)
    For $i \in [k]$, we denote by $v^i_1, \ldots, v^i_p$ the vertices of $V_i$. 
    We first describe the reduction of Fomin et al.~\cite{FGR17} and then explain how to modify it to prove the theorem.
    
    \medskip
    \noindent{\bf The Construction of Fomin, Golovach and Raymond~\cite{FGR17}.} The graph $H$ is obtained as follows.
    \begin{enumerate}[1.]
        \item Construct $k$ nodes $u_1, \ldots, u_k$.\label{construction:H:vertices}
        \item For every $1 \le i < j \le k$, construct a node $w_{i, j}$ and two pairs of parallel edges $u_i w_{i, j}$ and $u_j w_{i, j}$.\label{construction:H:edges}
    \end{enumerate}
    We then construct the subdivision $H'$ of $H$ by first subdividing each edge $p$ times.
    We denote the subdivision nodes for $4$ edges of $H$ constructed for each pair $1 \le i < j \le k$ in Step~\ref{construction:H:edges} by 
$x_1^{(i, j)}, \ldots, x_p^{(i, j)}$, $y_1^{(i, j)}, \ldots, y_p^{(i, j)}$, $x_1^{(j, i)}, \ldots, x_p^{(j, i)}$, and $y_1^{(j, i)}, \ldots, y_p^{(j, i)}$.
    To simplify notation, we assume that $u_i = x_0^{(i ,j)} = y_0^{(i, j)}$, $u_j = x_0^{(j, i)} = y_0^{(j, i)}$ and $w_{i, j} = x^{(i, j)}_{p+1} = y^{(i, j)}_{p + 1} = x^{(j, i)}_{p+1} = y^{(j, i)}_{p + 1}$.
    
    We now construct the $H$-graph $G''$ by defining its $H$-representation $\cM = \{M_v\}_{v \in V(G'')}$ where each $M_v$ is a connected subset of $V(H')$. (Recall that $G$ denotes the graph of the {\sc Multicolored Clique} instance.)
    \begin{enumerate}
        \item For each $i \in [k]$ and $s \in [p]$, construct a vertex $z_s^i$ with model
            \begin{align*}
                M_{z_s^i} \defeq \bigcup\nolimits_{j \in [k], j \neq i}\left(\left\lbrace x_0^{(i, j)}, \ldots, x_{s - 1}^{(i, j)}\right\rbrace \cup \left\lbrace y_0^{(i, j)}, \ldots, y_{p-s}^{(i, j)}\right\rbrace \right).
            \end{align*}
        \item For each edge $v_s^i v_t^j \in E(G)$ for $s, t \in [p]$ and $1 \le i < j \le k$, construct a vertex $r_{s, t}^{(i, j)}$ with:
            \begin{align*}
                M_{r_{s, t}^{(i, j)}} \defeq &\left\lbrace x_s^{(i, j)},\ldots, x_{p+1}^{(i, j)}\right\rbrace 
                            \cup \left\lbrace y_{p-s+1}^{(i, j)},\ldots, y_{p+1}^{(i, j)}\right\rbrace \\
                            \cup &\left\lbrace x_t^{(j, i)}, \ldots, x_{p+1}^{(j, i)} \right\rbrace
                            \cup \left\lbrace y_{p-t+1}^{(j, i)},\ldots, y_{p+1}^{(j, i)}\right\rbrace.
            \end{align*}
    \end{enumerate}
    Throughout the following, for $i \in [k]$ and $1 \le i < j \le k$, respectively, we use the notation 
        \begin{align*}
            Z(i) \defeq \bigcup\nolimits_{s \in [p]}\left\lbrace z^i_s \right\rbrace  \mbox{ and } R(i, j) \defeq \bigcup\nolimits_{\substack{v^i_s v^j_t \in E(G), \\  s, t \in [p]}} \left\lbrace r_{s, t}^{(i, j)}\right\rbrace.
        \end{align*}
     We now observe the crucial property of $G''$.
    \begin{nestedobservation}[Claim 7 in~\cite{FGR17}]\label{obs:claim:7}
        For every $1 \le i < j \le k$, a vertex $z_h^i \in V(G')$ (a vertex $z_h^j \in V(G')$) is {\em not} adjacent to a vertex $r^{(i, j)}_{s, t} \in V(G')$ corresponding to the edge $v^i_s v^j_t \in E(G)$ if and only if $h = s$ ($h = t$, respectively).
    \end{nestedobservation}
\newcommand\bipgadget{\fB}
    \noindent{\bf The New Gadget.} We now describe how to obtain from $G''$ a graph $G'$ that will be the graph of the instance of \prob. We do so by adding a gadget to each set $Z(i)$ and $R(i, j)$ (for all $1 \le i < j \le k$). We first describe the gadget and then explain how to modify $H'$ to a new graph $K'$ such that $G'$ is a $K$-graph (where $K$ denotes the graph obtained from $K'$ by undoing the above described subdivisions that were made in $H$ to obtain $H'$).
    \begin{figure}
        \centering
        \includegraphics[height=.235\textheight]{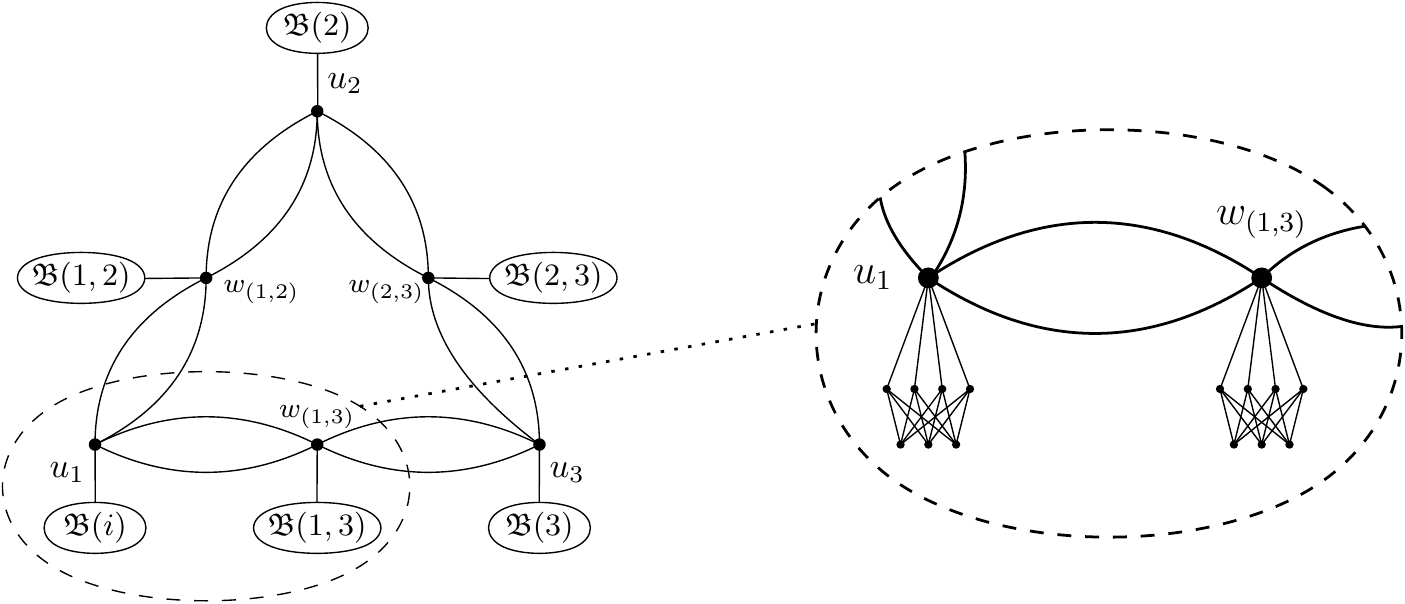}
        \caption{The graph $K$ with respect to which the graph $G'$ constructed in the proof of Theorem~\ref{thm:lb:max} is a $K$-graph. In this example, we have $k = 3$ and $d = 4$.}
        \label{fig:K}
    \end{figure}
    Let $X$ be any set of vertices of $G''$. The gadget $\bipgadget(X)$ is a complete bipartite graph on $2d-1$ vertices and bipartition $(\{\beta_{1,1}, \ldots, \beta_{1, d}\}, \{\beta_{2, 1}, \ldots, \beta_{2, d-1}\})$. such that for $h \in [d]$, each vertex $\beta_{1, h}$ is additionally adjacent to each vertex in $X$. For $1\le i < j \le k$, we use the notation $\bipgadget(i) \defeq \bipgadget(Z(i))$ and $\bipgadget(i, j) \defeq \bipgadget(R(i, j))$ and we denote their vertices by $\beta^i_{\cdot, \cdot}$ and $\beta^{(i, j)}_{\cdot, \cdot}$, respectively.
    
    We obtain $K'$ by `hardcoding' each gadget $\bipgadget(\cdot)$ into $H'$. That is, for $i \in [k]$, 
    we add the graph $\fB(i)$ and connect it to the remaining vertices via the edges
    $u_i \beta^i_{1, h}$ for $h \in [d]$.
    For $1 \le i < j \le k$, we proceed analogously in encoding $\bipgadget(i, j)$ into $H'$. For an illustration of the graph $K$, see Figure~\ref{fig:K}. We observe that $\card{K} = 2d \left(k + {k \choose 2}\right) = kd(k+1)$ and 
    \begin{align}
    \cardd{K} = 4 \binom{k}{2} + \left(k + \binom{k}{2}\right)\cdot(d + d(d-1)) = \frac{1}{2}k\left(d^2(k+1) + 4(k-1)\right).\label{eq:lb:max:num:edges:K}
    \end{align}
    We subdivide all newly introduced edges, i.e.\ all edges in $E(K') \setminus E(H')$ and for an edge $xy \in E(K') \setminus E(H')$, we denote the resulting vertex by $s(x, y)$.
    We are now ready to describe (the $K$-representation of) $G'$.
    \begin{enumerate}[1.]
        \item For all $i \in [k]$ and $s \in [p]$, we add the vertices $s(u_i, \beta^i_{1, h})$ (where $h \in [d]$) to the model of $z^i_s$. 
                For all $1 \le i < j \le k$ and $s, t \in [p]$ with $v^i_s v^j_t \in E(G)$, we add the vertices $s(w_{(i, j)}, \beta^{(i, j)}_{1, h})$ for $h \in [d]$ to the model of $r^{(i, j)}_{s, t}$.
        \item For all $i \in [k]$ and $h \in [d]$, we add a vertex $b^i_{1, h}$ with model 
                $\{\beta^i_{1, h}, s(u_i, \beta^i_{1, h})\} \cup \bigcup\nolimits_{h' \in [d-1]} \{s(\beta^i_{1, h}, \beta^i_{2, h'})\}$.
        \item For all $i \in [k]$ and $h \in [d-1]$, we add a vertex $b^i_{2, h}$ with model 
                $\{\beta^i_{2, h}\} \cup \bigcup\nolimits_{h' \in [d]} \{s(\beta^i_{2, h}, \beta^i_{1, h'})\}$.
        \item For all $v^i_s v^j_t \in E(G)$ (where $1 \le i < j \le k$ and $s ,t \in [p]$) and $h \in [d]$, we add a vertex $b^{(i, j)}_{1, h}$ with model 
                $\{\beta^{(i, j)}_{1, h}, s(w_{(i, j)}, \beta^{(i, j)}_{1, h})\} \cup \bigcup\nolimits_{h' \in [d-1]} \{s(\beta^{(i, j)}_{1, h}, \beta^{(i, j)}_{2, h'})\}$.
        \item For all $v^i_s v^j_t \in E(G)$ (where $1 \le i < j \le k$ and $s ,t \in [p]$) and $h \in [d-1]$, we add a vertex $b^{(i, j)}_{1, h}$ with model 
            $\{\beta^{(i, j)}_{2, h}\} \cup \bigcup_{h' \in [d]} \{s(\beta^{(i, j)}_{2, h}, \beta^{(i, j)}_{1, h'})\}$.
    \end{enumerate}
    \begin{figure}
       \centering
        \includegraphics[height=.235\textheight]{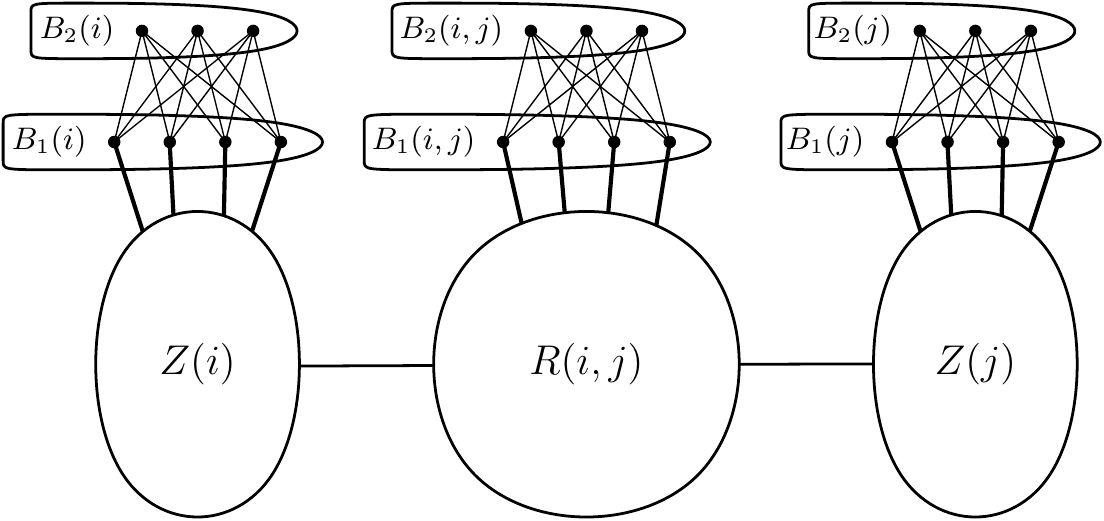}
        \caption{A part of the graph $G'$, where $1 \le i < j \le k$ and $d = 4$.}
        \label{fig:gprime}
    \end{figure}
    One can verify that these five steps introduce the above described vertices to $G'$. For an illustration of $G'$, see Figure~\ref{fig:gprime}.
	We now turn to the correctness proof of the reduction.
	
	For $1 \le i < j \le k$, we let $B_1(i) \defeq \{b^i_{1, 1}, \ldots, b^i_{1, d}\}$, $B_2(i) \defeq \{b^i_{2, 1}, \ldots, b^i_{2, d-1}\}$, $B_1(i, j) \defeq \{b^{(i, j)}_{1, 1}, \ldots, b^{(i, j)}_{1, d}\}$ and $B_2(i, j) \defeq \{b^{(i, j)}_{2, 1}, \ldots, b^{(i, j)}_{2, d-1}\}$; furthermore
    $B(i) \defeq B_1(i) \cup B_2(i)$, $B(i, j) \defeq B_1(i, j) \cup B_2(i, j)$, and $B \defeq \bigcup_{i \in [k]} B(i) \cup \bigcup_{1 \le i < j \le k} B(i, j)$. Note that $\card{B} = (2d-1)(k + \binom{k}{2})$. We furthermore let $k' \defeq 2d \cdot(k + \binom{k}{2})$ and throughout the following, we use the notation 
    \begin{align*}
        Z_{+B}(i) \defeq Z(i) \cup B(i) \mbox{ and }
        R_{+B}(i, j) \defeq R(i, j) \cup B(i, j).
    \end{align*}
    We now prove the first direction of the correctness of the reduction.
    Note that the following claim yields the forward direction of the correctness proof, since a \exactprobset ~set is a \probset ~set. (Recall that $d \in \rho^*$ and $x \le d+1$.)
    \begin{subclaim}\label{claim:lb:max:cor:forward}
        If $G$ has a multicolored clique, then $G'$ has a \exactprobset ~set of size $k' = 2d \cdot (k + \binom{k}{2})$ (assuming $k \ge 3$).
    \end{subclaim}
    \begin{claimproof}
        Let $\{v_{h_1}^1, \ldots, v_{h_k}^k\}$ be the vertex set in $G$ that induces the multicolored clique. By Observation~\ref{obs:claim:7} we can verify that
        \begin{align}
            I \defeq \left\lbrace z^1_{h_1}, \ldots, z^k_{h_k}\right\rbrace \cup \left\lbrace r^{(i, j)}_{h_i, h_j} \mid 1 \le i < j \le k \right\rbrace \label{eq:independent:set}
        \end{align}
        is an independent set in $G'$. We let $S \defeq I \cup B$ and observe that $S$ is a \probset ~set: By construction, there is no edge between any pair of distinct sets of $B(i)$, $B(i')$, $B(i, j)$, $B(i', j')$, for any choice of $1 \le i < j \le k$ and $1 \le i' < j' \le k$. 
        
        Consider any vertex $x \in S$ and suppose wlog.\footnote{The case when $x \in R_{+B}(i, j)$ can be argued for analogously.}~that $x \in Z_{+B}(i)$ for some $i \in [k]$. If $x = z^i_{h_i}$, then $x$ is adjacent to the $d$ vertices 
        $b^i_{1, 1}, \ldots, b^i_{1, d}$,    if 
        $x = b^i_{1, \ell}$ for some $\ell \in [d]$, then $x$ is adjacent to $z^i_{h_i}$ and the vertices 
        $b^i_{2, 1}, \ldots, b^i_{2, d-1}$ and if 
        $x = b^i_{2, \ell'}$ for some $\ell' \in [d-1]$, then it is adjacent to the vertices 
        $b^i_{1, 1}, \ldots, b^i_{1, d}$.
        Hence, in all cases, $x$ has precisely $d$ neighbors in $S$. 
        
        Let $y \in V(G') \setminus S$ and note that $(V(G') \setminus S) \cap B = \emptyset$. If $y \in Z(i)$ for some $i \in [k]$, then $N(y) \cap S \supseteq \{z^i_{h_i}, b^i_{1, 1}, \ldots, b^i_{1, d}\}$, so $\card{N(y) \cap S} \ge d+1$. Since the only additional neighbors of $y$ in $S$ are in the set $R_i \defeq \bigcup_{1 \le j < i} R(j, i) \cup \bigcup_{i < j \le k} R(i, j)$ and $R_i \cap S \subseteq I$, we can conclude that $\card{N(y) \cap (S \setminus B)} \le k-1$, since $I$ contains precisely one vertex from each set $R(i, j)$.
        We have argued that $d+1 \le \card{N(y) \cap S} \le d + k$. If $y \in R(i, j)$ for some $1 \le i < j \le k$, we can argue as before that $\card{N(y) \cap S} \ge d+1$ and since all neighbors of $y$ in $S \setminus B(i, j)$ are contained either in $Z(i)$ or $Z(j)$, we can conclude that $d+1 \le \card{N(y) \cap S} \le d+3 \le d+k$.
        
        It remains to count the size of $S$. Clearly, $\card{I} = k + \binom{k}{2}$ and as observed above, $\card{B} = (2d-1)(k + \binom{k}{2})$, so
        \[
            \card{S} = \card{I} + \card{B} = k + \binom{k}{2} + (2d-1)\left(k + \binom{k}{2}\right) = 2d\left(k + \binom{k}{2}\right) = k',
        \]
        as claimed.
    \end{claimproof}
    We now prove the backward direction of the correctness of the reduction. We begin by making several observations about the structure of \probset ~sets in the graph $G'$.
    \begin{subclaim}\label{claim:lb:max:cor:aux}
        Let $1 \le i < j \le k$. 
        \begin{enumerate}[(i)]
            \item Any \probset ~set in $G'$ contains at most $d+1$ vertices from each $Z(i) \cup B_1(i)$ or $R(i, j) \cup B_1(i, j)$.\label{claim:lb:max:cor:aux:1}
            \item Any \probset ~set contains at most $2d$ vertices from each $Z_{+B}(i)$ or $R_{+B}(i, j)$.\label{claim:lb:max:cor:aux:1b}
            \item If a \probset ~set $S$ contains $2d$ vertices from some $Z_{+B}(i)$ (resp., $R_{+B}(i, j)$), then it contains at least one vertex from $Z(i)$ (resp., $R(i, j)$) and each such vertex in $S \cap Z(i)$ (resp., $S \cap R(i, j)$) has at least $d$ neighbors in $S \cap Z_{+B}(i)$ (resp., $S \cap R_{+B}(i, j)$).\label{claim:lb:max:cor:aux:2}
        \end{enumerate} 
    \end{subclaim}
    \begin{claimproof}
        (\ref{claim:lb:max:cor:aux:1}) We prove the claim w.r.t.~a set $Z(i) \cup B_1(i)$ and remark that a proof for $R(i, j) \cup B_1(i, j)$ works analogously.
        Suppose not and let $S \subseteq V(G')$ be such that it contains at least $d+2$ vertices from some $Z(i) \cup B_1(i)$. Since $\card{B_1(i)} = d$, we know that $S$ contains a vertex from $Z(i)$, say $x$. However, by construction, all vertices in $S \cap (Z(i) \cup B_1(i)) \setminus \{x\}$ are adjacent to $x$, implying that $x$ has at least $d+1$ neighbors in $S$, a contradiction with the fact that $S$ is a \probset ~set. 
        
        (\ref{claim:lb:max:cor:aux:1b}) follows as a direct consequence, since $Z_{+B}(i) \setminus (Z(i) \cup B_1(i)) = B_2(i)$ and $\card{B_2(i)} = d-1$. Similar for $R_{+B}(i, j)$.
        
        For (\ref{claim:lb:max:cor:aux:2}), observe that if $S$ contains $2d$ vertices from $Z_{+B}(i)$, then it contains $B_2(i)$ and $d+1$ vertices from $Z(i) \cup B_1(i)$ by (\ref{claim:lb:max:cor:aux:1}) and the fact that $Z_{+B}(i) \setminus (Z(i) \cup B_1(i)) = B_2(i)$ and $\card{B_2(i)} = d-1$. Since $\card{B_1(i)} = d$, at least one vertex is in $S \cap Z(i)$. The claim now follows as any vertex in $Z(i)$ is adjacent to any other vertex in $Z(i)$ as well as any vertex in $B_1(i)$. Similar for $R_{+B}(i, j)$.
    \end{claimproof}
    \begin{subclaim}\label{claim:lb:max:cor:backward}
        If $G'$ contains a \probset ~set $S$ of size $k' = 2d(k + \binom{k}{2})$, then $G$ contains a multicolored clique.
    \end{subclaim}
    \begin{claimproof}
        Let $S$ be a \probset ~set of size $k'$ in $G'$. 
%
        By Claim~\ref{claim:lb:max:cor:aux}(\ref{claim:lb:max:cor:aux:1b}), we can conclude that $S$ contains precisely $2d$ vertices from each $Z_{+B}(i)$ and each $R_{+B}(i, j)$ (where $1 \le i < j \le k$). Consider any pair $i$, $j$ with $1 \le i < j \le k$. By Claim~\ref{claim:lb:max:cor:aux}(\ref{claim:lb:max:cor:aux:2}) we know that there are vertices 
        \[
            z_{s_i}^i \in Z(i) \cap S, \mbox{~~} z_{s_j}^j \in Z(j) \cap S, \mbox{~~and~} r_{t_i, t_j}^{(i, j)} \in R(i, j) \cap S.
        \]
    Again by Claim~\ref{claim:lb:max:cor:aux}(\ref{claim:lb:max:cor:aux:2}), we can conclude that $z_{s_i}^i    r_{t_i, t_j}^{(i, j)} \notin E(G')$ and $z_{s_j}^j r_{t_i, t_j}^{(i, j)} \notin E(G')$: E.g., $z_{s_i}^i$ has $d$ neighbors in $Z_{+B}(i) \cap S$, so if $z_{s_i}^i     r_{t_i, t_j}^{(i, j)} \in E(G')$, then $z_{s_i}^i$ has $d+1$ neighbors in $S$, a contradiction with the fact that $S$ is a \probset ~set. By Observation~\ref{obs:claim:7}, we then have that $s_i = t_i$ and $s_j = t_j$. We can conclude that $v_{h_i}^i v_{h_j}^j \in E(G)$ and since the argument holds for any pair of indices $i, j$ that $G$ has a multicolored clique.
    \end{claimproof}
    We would like to remark that by the proof of the previous claim, we have established that any \probset ~set $S$ in $G'$ of size $k'$ in fact contains all vertices from $B$ and one vertex from each $Z(i)$ and from each $R(i, j)$. Since this is precisely the shape of the set constructed in the forward direction of the correctness proof, this shows that any \probset ~set of size $k'$ in $G'$ is a \exactprobset ~set (assuming $k \ge 3$).
    
    Claims~\ref{claim:lb:max:cor:forward} and~\ref{claim:lb:max:cor:backward} establish the correctness of the reduction
    We observe that $\card{V(G')} = \cO(\card{V(G)} + d^2 \cdot k^2)$ and clearly, $G'$ can be constructed from $G$ in time polynomial in $\card{V(G)}$, $d$ and $k$ as well. Furthermore, by (\ref{eq:lb:max:num:edges:K}), $\cardd{K} = \cO(d^2 \cdot k^2)$ and the theorem follows.
\end{proof}
By Proposition~\ref{thm:h:graph:mim}, the previous theorem implies
\begin{corollary}
	For any fixed $d \in \bN$ and $x \le d+1$, the following holds. Let $\sigma^* \subseteq \bN_{\le d}$ with $d \in \sigma^*$. Then, \prob ~is \Wone{}-hard parameterized by linear mim-width plus solution size, and the hardness holds even if a corresponding decomposition tree is given.
\end{corollary}
    \renewcommand\probset{$(\sigma^*, \rho^*)$}
    \renewcommand\prob{\minprob -\probset {\sc ~Domination}}

\subsection{Minimization Problems}
In this section we prove hardness of minimization versions of several $(\sigma, \rho)$ problems. We obtain our results by modifying a reduction from {\sc Multicolored Independent Set} to {\sc Dominating Set} due to Fomin et al.~\cite{FGR17}. In the {\sc Multicolored Independent Set} problem we are given a graph $G$ and a partition $V_1, \ldots, V_k$ of its vertex set $V(G)$ and the question is whether there is an independent set $\{v_1, \ldots, v_k\} \subseteq V(G)$ in $G$ such that for each $i \in [k]$, $v_i \in V_i$. The \Wone{}-hardness of this problem follows immediately from the \Wone{}-hardness of the {\sc Multicolored Clique} problem.

\medskip
\noindent{\bf The Reduction of Fomin et al.~\cite{FGR17}.} Let $G$ be an instance of {\sc Multicolored Independent Set} with partition $V_1, \ldots, V_k$ of $V(G)$. Again we can assume that $k \ge 2$ and that $\card{V_i} = p$ for all $i \in [k]$. If the latter condition does not hold, let $p \defeq \max_{i \in [k]} \card{V_i}$ and for each $i \in [k]$, add $p - \card{V_i}$ vertices to $V_i$ that are adjacent to all vertices in each $V_j$ where $j \neq i$. It is clear that the resulting instance has a multicolored independent set if and only if the original instance does.

The graph $G'$ of the {\sc Minimum Dominating Set} instance is obtained from the graph $G''$ as constructed in the proof of Theorem~\ref{thm:lb:max}.\footnote{See the paragraph `The Construction of Fomin, Golovach and Raymond'.} The only difference is that for $i \in [k]$, a vertex $b_i$ is added whose model is $\{u_i\}$, i.e. it is adjacent to all vertices in $Z(i)$ and nothing else. 
We argue that $G$ has a multicolored independent set if and only if $G'$ has a dominating set of size $k$.

For the forward direction, if $G$ has a multicolored independent set $I \defeq \lbrace v^1_{h_1}, \ldots, v^k_{h_k}\rbrace$, then using Observation~\ref{obs:claim:7}, one can verify that $D \defeq \lbrace z^1_{h_1}, \ldots, z^k_{h_k}\rbrace$ is a dominating set in $G'$: Clearly, for each $i \in [k]$, the vertices in $Z(i) \cup \{b_i\}$ are dominated by $z^i_{h_i} \in D$. Suppose there is a vertex $r^{(i, j)}_{s, t} \in R(i, j)$ that is not dominated by $D$, then in particular it is neither adjacent to $z^i_{h_i}$ nor to $z^j_{h_j}$. By Observation~\ref{obs:claim:7}, this implies that $G$ contains the edge $v^i_{h_i} v^j_{h_j}$, a contradiction with the fact that $I$ is an independent set.

For the backward direction, suppose that $G'$ has a dominating set $D$ of size $k$. Due to the vertices $b_i$ (for $i \in [k]$), we can conclude that for all $i \in [k]$, $D \cap (Z(i) \cup \{b_i\}) \neq \emptyset$. If $D$ contains $b_i$ for some $i \in [k]$, then we can replace $b_i$ by any vertex in $Z(i)$ such that the resulting set is still a dominating set of $D$, so we can assume that $D = \{z^1_{h_1}, \ldots, z^k_{h_k}\}$. We claim that $\{v^1_{h_1}, \ldots, v^k_{h_k}\}$ is an independent set in $G$. Suppose that for $i, j \in [k]$, there is an edge $v^i_{h_i} v^j_{h_j} \in E(G)$. Observation~\ref{obs:claim:7} implies that $r^{(i, j)}_{h_i, h_j}$ is neither adjacent to $z^{i}_{h_i}$ nor to $z^j_{h_j}$, so $r^{(i, j)}_{h_i, h_j}$ is not dominated by $D$, a contradiction.
    \renewcommand\probset{$(\sigma^*, \rho^*)$}
    \renewcommand\prob{\minprob -\probset {\sc ~Domination}}
\begin{remarknr}\label{rem:ds:hardness}   
    We would like to remark that the above reduction is to the \prob ~problem, for all $\sigma^* \subseteq \bN$ with $0 \in \sigma^*$ and $\rho^* \subseteq \bN^+$ with $\{1, 2\} \subseteq \rho^*$.
\end{remarknr}

\medskip
\noindent{\bf Adaption to Total Domination Problems.} Recall that the $(\sigma, \rho)$-formulation for {\sc Dominating Set} is $(\bN, \bN^+)$. We now explain how to modify the above reduction to obtain $\Wone$-hardness for dominating set problems where each vertex in the solution has to have at least one neighbor in the solution. These problems include {\sc Total Dominating Set} and {\sc Dominating Induced Matching}, which can be formulated as $(\bN^+, \bN^+)$ and $(\{1\}, \bN^+)$, respectively. The minimization problem of either of them is known to be NP-complete.
\begin{theorem}\label{cor:total:ds}
	For $\sigma^* \subseteq \bN^+$ with $1 \in \sigma^*$ and $\rho^* \subseteq \bN^+$ with $\{1, 2\} \subseteq \rho^*$, \prob ~is \Wone{}-hard on $H$-graphs parameterized by the number of edges in $H$ plus solution size, and the hardness holds even when an $H$-representation of the input graph is given.
\end{theorem}
\begin{proof}
    We modify the above reduction from {\sc Multicolored Independent Set} as follows. For each $i \in [k]$, we add another vertex $c_i$ to $G'$ which is only adjacent to $b_i$. We let $B \defeq \bigcup_{i \in [k]} \{b_i\}$ and $C \defeq \bigcup_{i \in [k]} \{c_i\}$. Note that these new vertices can be hardcoded into $H$ with the number of edges in $H$ increasing only by $k$. 
    To argue the correctness of the reduction, we now show that $G$ has a multicolored independent set if and only if $G'$ has a \probset ~set of size $k' \defeq 2k$.
    
    For the forward direction, suppose that $G$ has an independent set $\{v^1_{h_1}, \ldots, v^k_{h_k}\}$. Then, $D' \defeq \{z^1_{h_1}, \ldots, z^k_{h_k}\}$ dominates all vertices in $V(G') \setminus C$ by the same argument as above and $D \defeq D' \cup B$ dominates all vertices of $G'$. Furthermore, each $x \in D$ has precisely one neighbor in $D$: For each such $x$, either $x = z^i_{h_i}$ or $x = b_i$ for some $i \in [k]$. In the former case, $N(x) \cap D = \{b_i\}$ and in the latter case, $N(x) \cap D = \{z^i_{h_i}\}$. Now let $y \in V(G') \setminus D$. If $y \in Z(i) \cup \{c_i\}$ for $i \in [k]$, then $N(y) \cap D = \{z^i_{h_i}, b_i\}$. If $y \in R(i, j)$ for some $1 \le i < j \le k$, then $y$ is either dominated by one of $z^i_{h_i}$ and $z^j_{h_j}$ or by both and it cannot have any other neighbors in $D$ by construction. Since $1 \in \sigma^*$ and $\{1, 2\} \subseteq \rho^*$, $D$ is a \probset ~set and clearly, $\card{D} = 2k$.
    
    For the backward direction, suppose that $G'$ has a \probset ~set $D$ of size $2k$. Let $i \in [k]$. Since $0 \notin \sigma^*$ and $0 \notin \rho^*$, we have that either $c_i \in D$ or $b_i \in D$ (either $c_i$ is dominating or it needs to be dominated). Since $c_i$ does not dominate any vertex in $G'$ other than $b_i$ and $b_i$ dominates $c_i$ plus all vertices in $Z(i)$, we can always assume that $b_i \in D$ and hence that $B \subseteq D$. Since $0 \notin \sigma^*$, all vertices of $B$ have a neighbor in $D$. For each $i \in [k]$, we can assume that this neighbor is some $z^i_{h_i}$ (rather than $c_i$, for similar reasoning as above). We have that $D = B \cup \{z^1_{h_1}, \ldots, z^k_{h_k}\}$ and since $D$ is a dominating set (in other words, $0 \notin \rho^*$), we can again argue using Observation~\ref{obs:claim:7} that $\{v^1_{h_1}, \ldots, v^k_{h_k}\}$ is an independent set in $G$.
\end{proof}
As a somewhat orthogonal result to Theorem~\ref{thm:lb:max}, we now show hardness of several problems related to the {\sc $d$-Dominating Set} problem, where each vertex that is not in the solution set has to be dominated by at least some fixed number of $d$ neighbors in the solution.

\medskip
\noindent{\bf Adaption to $d$-Domination Problems.} We use a similar gadget constructed in the proof of Theorem~\ref{thm:lb:max} to prove $\Wone$-hardness of several $(\sigma, \rho)$ problems where each vertex has to be dominated by at least $d$ vertices. In particular, we prove the following theorem. Note that the analogous statement of the following theorem for $d = 1$ is proved by the reduction explained in the beginning of this section, see Remark~\ref{rem:ds:hardness}.
\begin{theorem}\label{thm:lb:min:d:dom}
    For any fixed $d \in \bN_{\geq 2}$,\footnote{Note that the analogous statement for $d=1$ follows from the reduction given in~\cite{FGR17}.} the following holds. Let $\sigma^* \subseteq \bN$ with $\{0, 1, d-1\} \subseteq \sigma^*$ and $\rho^* \subseteq \bN_{\ge d}$ with $\{d, d+1\} \subseteq \rho^*$. Then, \prob ~is $\Wone$-hard on $H$-graphs parameterized by the number of edges in $H$ plus solution size, and the hardness even holds when an $H$-representation of the input graph is given.
\end{theorem}
\begin{proof}
    We modify the above reduction from {\sc Multicolored Independent Set}. Let $G$ be a graph with vertex partition $V_1, \ldots, V_k$ and $\card{V_i} = p$ for all $i \in [k]$ and assume $k \ge 2$. We first describe the gadget we use and then how to construct the graph $G'$ of the \prob ~instance.

\medskip    
\noindent{\bf The Gadget $\fC(i)$.} Let $i \in [k]$. The gadget $\fC(i)$ is a complete bipartite graph on bipartition $(C_1(i), C_2(i))$ with $C_1(i) \defeq \{c^i_{1, 1}, \ldots, c^i_{1, d}\}$ and $C_2(i) \defeq \{c^i_{2, 1}, \ldots, c^i_{2, d}\}$ such that each vertex $c^i_{1, j}$ for $j \in [d-1]$ is additionally adjacent to all vertices in $Z(i)$ as well as to all vertices in $R(i, j)$ for $j > i$.
(Note that $c^i_{1, d}$ does not have these additional adjacencies.) Throughout the following, we let $C(i) \defeq C_1(i) \cup C_2(i)$ and $C \defeq \bigcup_{i \in [k]} C(i)$.

    \begin{figure}
        \centering
        \includegraphics[height=.235\textheight]{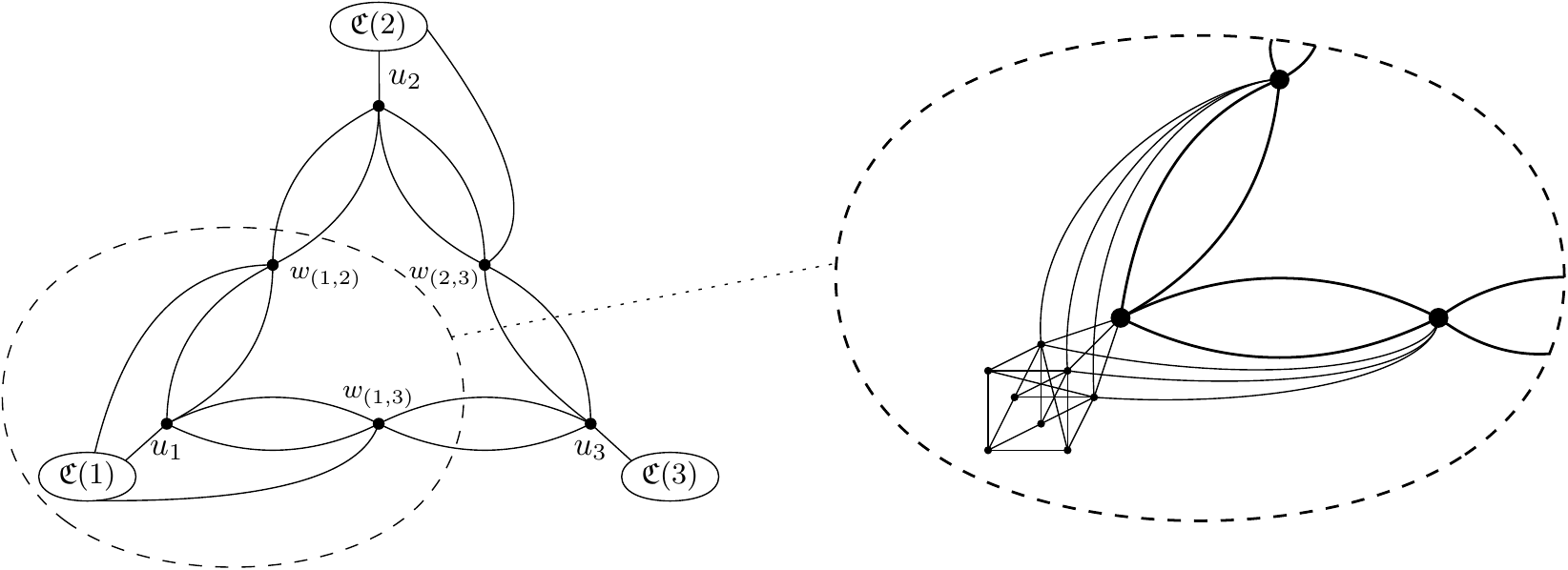}
        \caption{An example graph $K$ w.r.t.~which the graph $G'$ constructed in the proof of Theorem~\ref{thm:lb:min:d:dom} is a $K$-graph. In this example, $k = 3$.}
        \label{fig:k:min}
    \end{figure}
    The graph $G'$ is now obtained by constructing the graph $G''$ as in the proof of Theorem~\ref{thm:lb:max} and then, for each $i \in [k]$, adding the gadget $\fC(i)$ and adding a `satellite vertex' $s_i$, adjacent to all vertices in $Z(i) \cup C_1(i)$. $G'$ is a $K$-graph for the graph $K \supseteq H$, obtained by `hardcoding' each $\fC(i)$, for $i \in [k]$, into $H$. That is, for each $i \in [k]$, we add a complete bipartite graph with bipartition $(\{\gamma^i_{1, 1}, \ldots, \gamma^i_{1, d}\}, \{\gamma^i_{2, 1}, \ldots, \gamma^i_{2, d}\})$, and make all vertices $\gamma^i_{1, h}$, where $h \in [d-1]$, adjacent to $u_i$ as well as to all vertices $w_{(i, j)}$ with $j > i$. For an illustration of $K$ see Figure~\ref{fig:k:min}. Note that
    \begin{align}
        \cardd{K} = \cardd{H} + kd^2 + \sum_{i = 1}^k (k - i)(d-1) = \cO(k^2 \cdot d + k \cdot d^2) \label{eq:d:dom:num:edges:k}
    \end{align}

One can now argue that $G'$ is a $K$-graph. Since the construction is completely analogous to that explained in the proof of Theorem~\ref{thm:lb:max}, we skip the details here. We illustrate the structure of the graph $G'$ in Figure~\ref{fig:gprime:min}.
    \begin{figure}
        \centering
        \includegraphics[height=.265\textheight]{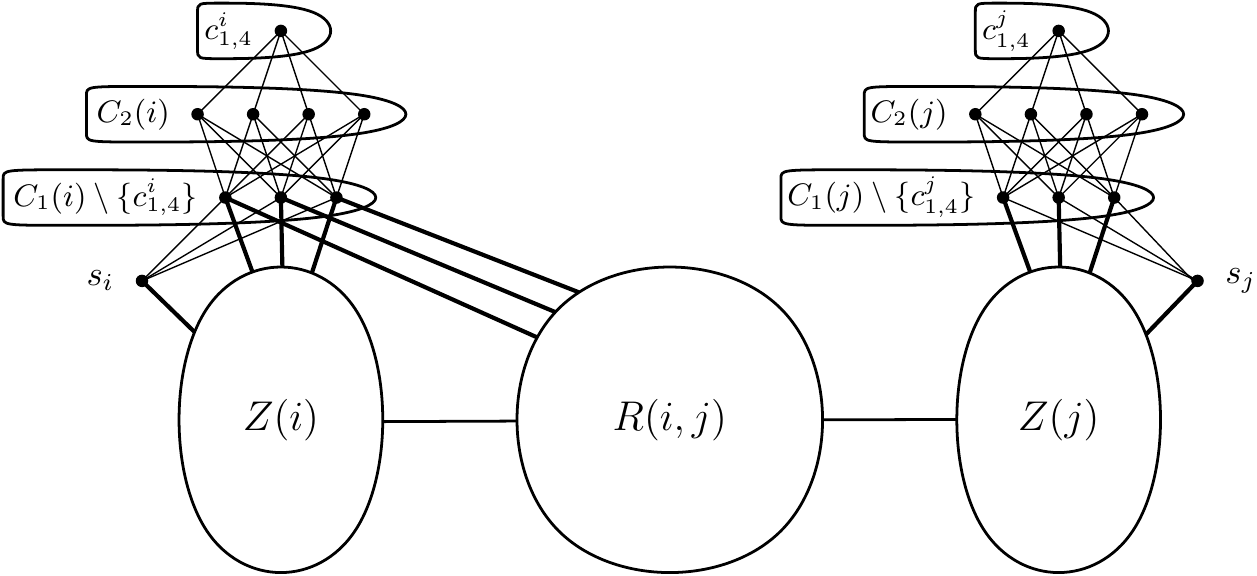}
        \caption{Illustration of a part of $G'$ constructed in the proof of Theorem~\ref{thm:lb:min:d:dom}, where $1 \le i < j \le k$ and $d = 4$.}
        \label{fig:gprime:min}
    \end{figure}
\begin{subclaim}\label{claim:d:dom:cor:forward}
    If $G$ has a multicolored independent set, then $G'$ has a \probset ~set of size $k' \defeq k \cdot (d+1)$.
\end{subclaim}
\begin{claimproof}
    Let $\{v^1_{h_1}, \ldots, v^k_{h_k}\}$ be the independent set in $G$. By the above reduction, $D' \defeq \{z^1_{h_1}$, $\ldots$, $z^k_{h_k}\}$ is a $(\{0\}, \{1, 2\})$-set of $G' - C$ (see also Remark~\ref{rem:ds:hardness}) of size $k$. Let $C_1 \defeq \bigcup_{i \in [k]} C_1(i)$, $C_2 \defeq C \setminus C_1$ and $D \defeq D' \cup C_1$. 
    
    Since each vertex in $V(G') \setminus (D \cup C)$ is adjacent to precisely $d-1$ vertices in $C_1$ and to either one or two vertices in $D'$ (and $D' \cap C_1 = \emptyset$), we can conclude that each vertex in $V(G') \setminus (D \cup C)$ is adjacent to either $d$ or $d+1$ vertices in $D$. Since each $C(i)$ induces a $K_{d, d}$, we can conclude that all vertices in $C_2$ have $d$ neighbors in $D$ as well. Furthermore, $N(s_i) \cap D = (C_1(i) \setminus \{c^i_{1, d}\}) \cup \{z^i_{h_i}\}$, so we have that all vertices in $G'$ that are not contained in $D$ have either $d$ or $d+1$ neighbors in $D$.
    
    Let $i \in [k]$. Then, $N(z^i_{h_i}) \cap D = \{c^i_{1, 1}$, $\ldots$, $c^i_{1, d-1}\}$,
    $N(c^i_{1, d}) \cap D = \emptyset$ and for $\ell \in [d-1]$, $N(c^i_{1, \ell}) \cap D = \{z^i_{h_i}\}$. We can conclude that $D$ is a $(\{0, 1, d-1\}, \{d, d+1\})$-set in $G'$ and clearly, $\card{D} = k + kd = k'$.
\end{claimproof}

In what follows, the strategy is to argue that each \probset ~set of size $k'$ contains a set $\{z^1_{h_1}, \ldots, z^k_{h_k}\}$ which will imply that $\{v^1_{h_1}, \ldots, v^k_{h_k}\}$ is an independent set in $G$.

\begin{subclaim}\label{claim:lb:d:dom:dplusone:z}
    For all $i \in [k]$, any \probset ~set $D$ in $G'$ contains at least $d$ vertices from $C(i)$ and at least $d+1$ vertices from $Z_+(i)$.
\end{subclaim} 
\begin{claimproof}
    We first show that each such $D$ contains at least $d$ vertices from $C(i)$. Suppose not, then $\card{D \cap C(i)} \le d-1$ for some $i \in [k]$. If $c^i_{1, d} \notin D$, then $C_2(i) \subseteq D$, otherwise $c^i_{1, d}$ cannot have $d$ or more neighbors in $D$. But $\card{C_2(i)} = d$, a contradiction. We can assume that $c^i_{1, d} \in D$. Furthermore, there is at least one vertex $c^i_{2, \ell}$ for $\ell \in [d]$ with $c^i_{2, \ell} \notin D$. To ensure that $c^i_{2, \ell}$ has at least $d$ neighbors in $D$, we would have to include all remaining vertices from $C_1(i)$ in $D$, but then $\card{D \cap C(i)} \ge d$, a contradiction. The claim now follows since the vertex $s_i$ only has neighbors in $Z_+(i)$ and at most $d-1$ neighbors in $D$ (namely $C_1(i) \setminus \{c^i_{1, d}\}$): Since $D$ is a \probset ~set, it either has to contain $s_i$ or at least one additional neighbor of $s_i$.
\end{claimproof}

\begin{subclaim}\label{claim:lb:d:dom:structure:ds}
    For all $i \in [k]$, any \probset ~set $D$ of size at most $k' = k(d+1)$ contains $C_1(i)$. We furthermore can assume that it additionally contains some $z^i_{h_i} \in Z(i)$, where $h_i \in [p]$.
\end{subclaim}
\begin{claimproof}
    By Claim~\ref{claim:lb:d:dom:dplusone:z} we have that $D$ contains $d+1$ vertices from each $Z_+(i')$, $i' \in [k]$, and no other vertices. Consider any vertex $z^i_s \in Z(i)$ (where $s \in [p]$) that is not contained in $D$. 
    Recall that $z^i_s$ has to have at least $d$ neighbors in $D$. By Clam~\ref{claim:lb:d:dom:dplusone:z}, $z^i_s$ has precisely one neighbor in $(Z(i) \cup \{s_i\}) \cap D$ and since $D$ does not contain any vertex from any $R(j, i)$ ($1 \le j < i$) or $R(i, j')$ ($i < j' \le k$), the only possible neighbors of $z^i_s$ in $D$ are $C_1(i) \setminus \{c^i_{1, d}\}$. We can conclude that $C_1(i) \subseteq D$. Now suppose that $s_i \in D$. Then, after swapping $s_i$ with any vertex in $Z(i)$, the resulting set remains a \probset ~set, and the claim follows.
\end{claimproof}

We are now ready to conclude the correctness proof of the reduction.

\begin{subclaim}\label{claim:d:dom:cor:backward}
    If $G'$ has a \probset ~set of size $k' = k(d+1)$, then $G$ has a multicolored independent set.
\end{subclaim}
\begin{claimproof}
    Let $D$ be a \probset ~set of size $k'$. By Claim~\ref{claim:lb:d:dom:structure:ds}, we can assume that $D = C_1 \cup \{z^1_{h_1}, \ldots, z^k_{h_k}\}$ for some $h_1, \ldots, h_k \in [p]$. Now, since for each $1 \le i < j \le k$, all vertices in $R(i ,j)$ have precisely $d - 1$ neighbors in $C_1$, each of them has to have at least one of $z^i_{h_i}$ and $z^j_{h_j}$ as a neighbor. By Observation~\ref{obs:claim:7}, this allows us to conclude that $\{v^1_{h_1}, \ldots, v^k_{h_k}\}$ is an independent set in $G$.
\end{claimproof}
    Claims~\ref{claim:d:dom:cor:forward} and~\ref{claim:d:dom:cor:backward} establish the correctness of the reduction. Clearly, $\card{V(G')} = \cO(\card{V(G)} + d^2 \cdot k)$ (and $G'$ can be constructed in polynomial time) and by (\ref{eq:d:dom:num:edges:k}), $\cardd{K} = \cO(k^2 \cdot d + k \cdot d^2)$. The theorem follows.
\end{proof}
Similarly to above, a combination of the previous two theorems with Proposition~\ref{thm:h:graph:mim} yields the following hardness results for $(\sigma, \rho)$ mimization problems on graphs of bounded linear mim-width.
\begin{corollary}
    Let $\sigma^* \subseteq \bN$ and $\rho^* \subseteq \bN$. Then, \prob ~is \Wone{}-hard parameterized by linear mim-width plus solution size, if one of the following holds.
    \begin{enumerate}[(i)]
        \item $\sigma^* \subseteq \bN^+$ with $1 \in \sigma^*$ and $\rho^* \subseteq \bN^+$ with $\{1, 2\} \subseteq \rho^*$.
        \item For some fixed $d \in \bN_{\ge 2}$, $\{0, 1, d-1\} \subseteq \sigma^*$ and $\rho^* \subseteq \bN_{\ge d}$ with $\{d, d+1\} \subseteq \rho^*$.
    \end{enumerate}
    Furthermore, the hardness holds even if a corresponding decomposition tree is given.
\end{corollary}

\section{Concluding Remarks}\label{sec:conclusion}
We have introduced the class of distance-$r$ $(\sigma, \rho)$ and \LCVPshort{} problems. This generalizes well-known graph distance problems
like distance-$r$ domination, distance-$r$ independence, distance-$r$ coloring and perfect $r$-codes. It also introduces many new  distance problems for which the standard distance-$1$ version naturally captures a well-known graph property. 

Using the graph parameter mim-width, we showed that all these problems are solvable in polynomial time for many interesting graph classes. These meta-algorithms will have runtimes which can likely be improved significantly for a particular problem on a particular graph class. 
For instance, blindly applying our results to solve {\sc Distance-$r$ Dominating Set} on permutation graphs yields an algorithm that runs in time $\cO(n^8)$: Permutation graphs have linear mim-width $1$ (with a corresponding decomposition tree that can be computed in linear time) \cite[Lemmas 2 and 5]{BELMONTE201354}, so we can apply Corollary~\ref{cor:dist:r:mim}(\ref{cor:dist:r:mim:lin}).
However, there is an algorithm that solves {\sc Distance-$r$ Dominating Set} on permutation graphs in time $\cO(n^2)$~\cite{RPP16}; a much faster runtime.
A concrete example of improving a mim-width based algorithm on a specific graph class has recently been provided by Chiarelli et al.~\cite{CHL18} who gave algorithms for the {\sc (Total) $k$-Dominating Set} problems that run in time $\cO(n^{3k})$ on proper interval graphs.
The fastest previously known algorithm runs in time $\cO(n^{4 + 6k})$~\cite{BELMONTE201354,BUIXUAN201366}, the generic mim-width based algorithm (cf.~Proposition~\ref{prop:sigma:rho-algorithm-nondistance-minmw} and~\cite[Lemmas 2 and 3]{BELMONTE201354})

We would like to draw attention to the most important and previously stated~\cite{JKT18,DBLP:journals/tcs/SaetherV16,vatshelle2012new} open question regarding the mim-width parameter:
Is there an \XP{} approximation algorithm for computing mim-width? An important first step could be to devise a polynomial-time algorithm deciding if a graph has mim-width 1, or even linear mim-width 1.

Regarding lower bounds, we expanded on the previous results by Fomin et al.~\cite{FGR17} and showed that many $(\sigma, \rho)$ problems are \Wone{}-hard parameterized by mim-width. However, it remains open whether there exists a problem which is \NP{}-hard in general, yet \FPT{} by mim-width. In particular, there are currently no hardness results when $\sigma$ and $\rho$ are both finite. Even so, we conjecture that every \NP{}-hard (distance) $(\sigma, \rho)$ problem is \Wone{}-hard parameterized by mim-width.

\bibliographystyle{plain}
\bibliography{ref}

\newpage

\appendix
\section{Basic definitions and notation}
We let the set of natural numbers be $\bN = \{0, 1, 2, \ldots\}$, and the positive natural numbers be $\bN^+ = \bN \setminus \{0\}$. For a set $S$ and a given property $\psi$, we denote by $S_\psi$ the biggest subset of $S$ where $\psi$ is satisfied for all elements. For instance, $\bN^+_{\leq k}$ denotes the set $\{1, 2, \ldots k\}$. For this particular property, we also use the shorthand $[k] = \bN^+_{\leq k}$.

A set $A \subseteq \bN$ is finite if it has finite cardinality, and it is co-finite if $\bN \setminus A$ has finite cardinality.


\medskip
\noindent{\bf Graphs.}
For a graph $G$ and a vertex $u \in V(G)$, its \emph{neighborhood} $N(u)$ is the set of all vertices adjacent to $u$. The \emph{closed neighborhood} of $u$ is denoted $N[u] = N(u) \cup \{v\}$.
The \emph{degree} of a vertex is the number of vertices adjacent to $v$ in the graph, $deg(v) = |N(v)|$. A vertex of degree $1$ is called a leaf, and the set of all leaves of a graph $G$ is denoted $L(G)$. 
In cases where it would otherwise be unclear which graph is being referred to, a subscript is added, \eg{} $N_{G'}(u)$ denotes the neighborhood of $u$ in the graph $G'$.

For two vertices $u,v \in V(G)$, the \emph{distance} between them $\dist(u, v)$ is the shortest possible length of a path with $u$ and $v$ as its endpoints, or $\infty$ if no path exist. A graph is \emph{connected} if $\dist(u,v) < \infty$ for all vertices $u, v \in V(G)$. For a positive integer $r$, the \emph{$r$-neighbourhood} of a vertex $u$ is the set of vertices at distance $r$ or less from $u$, denoted $N^r(u) = \{v \in V(G) \setminus \{u\} \mid \dist(u, v) \leq r\}$.

A \emph{connected component} is a vertex maximal induced subgraph which is connected. A \emph{tree} $T$ is a connected graph which contains no cycles. A \emph{caterpillar} is a graph which consists of a path, and for each non-endpoint vertex of the path there is an additional leaf attached to that vertex.

We denote by $n = |G| = |V(G)|$ the number of vertices, and by $m = ||G|| = |E(G)|$ the number of edges of a graph $G$.

\medskip
\noindent{\bf Parameterized Complexity Theory.}
A {\em parameterized problem} is a problem where the input instances come along with an non-negative integer $k$, the {\em parameter}. Formally, it is a language $L \subseteq \Sigma^* \times \bN$, where $\Sigma$ is a fixed, finite alphabet. The parameter $k$ is sometimes given implicitly.

A parameterized problem is in the class FPT if there exsists an algorihtm which correctly decides every instance $(x, k)$ of the problem in time $f(k) \cdot |(x, k)|^c$ for some constant $c$. Such an algorithm is called and FPT algorithm.

A parameterized problem is in the class \XP{} if for every instance $(x, k)$ there exsists an algorithm which solves it in time $f(k) \cdot |(x, k)|^{g(k)}$.

For both FPT and XP algorithms, the runtime becomes polynomial when the parameter is bounded. But while it is clear that $\classFPT \subseteq \XP$, the converse is not true under basic complexity assumptions, and there is a hierarchy of complexity classes between them called the W hieararchy: $\classFPT \subseteq \Wone \subseteq \W[2] \ldots \subseteq \XP$. We say that a problem is $\Wone$-hard if every problem in $\Wone$ can be reduced to that problem by a parameter-preserving reduction. For a more thorough introduction, we refer the reader to~\cite{cygan2015parameterized}.


\end{document}